\documentclass[11pt]{article}

\usepackage{amsmath,amssymb,amsthm}
\usepackage{booktabs}
\usepackage{array}
\usepackage{algorithm}
\usepackage[noend]{algpseudocode}
\usepackage{hyperref}
\usepackage{geometry}
\usepackage{url}
\newtheorem{lemma}{Lemma}
\newtheorem{proposition}{Proposition}
\newtheorem{corollary}{Corollary}

\DeclareMathOperator{\round}{round}
\DeclareMathOperator*{\argmin}{arg\,min}

\newcommand{\Z}{\mathbb{Z}}
\newcommand{\R}{\mathbb{R}}
\newcommand{\one}{\mathbf{1}}

\title{\bf Faster Closest-Point Algorithms\\ for the $E_6^*$ and $E_7^*$ Lattices}
\author{Yuriy Reznik\\[3pt]
        \normalsize Massachusetts Institute of Technology\\
        \normalsize \texttt{yreznik@mit.edu}}
\date{}

\begin{document}
\maketitle

\begin{abstract}
The dual lattices $E_6^*$ and $E_7^*$ are of particular interest in source coding and data compression applications. Among all known lattices in 
dimensions six and seven they attain the smallest normalized second moments, i.e., the
smallest average quantization error. Their use in practice requires fast
closest-point (nearest-lattice-point) algorithms. The known approach,
due to Conway and Sloane and completed for $E_6$ and $E_6^*$ by
Takizawa, Yagi, and Kawabata (TYK), decodes these lattices as unions of
cosets of root lattices $A_n$: each coset is decoded separately, and
the best result is kept. This requires four coset decodings for
$E_7^*$ and six for $E_6^*$, together with explicit distance
computations.

This paper shows that all these coset decodings can be collapsed into a
single sweep. Reformulated in terms of glue vectors, the TYK
decompositions state that $E_7^*$ is the union of the even glue classes
of $A_7^*$, and that $E_6^*$ is a parity-matched sublattice of
$A_1^*\oplus A_5^*$. The candidate chain constructed by the
closest-point algorithm of McKilliam, Clarkson, and Quinn (MCQ) for
$A_n^*$ visits every glue class of $A_n$ exactly once and is optimal
within each class. Consequently, one sorted sweep per coordinate block
yields the closest points of all glue cosets simultaneously, and
$E_6^*$ and $E_7^*$ are decoded at roughly the cost of a single $A_5^*$
or $A_7^*$ quantization. Rough operation counts indicate a
$4$--$6\times$ reduction for $E_6^*$ and $3$--$4\times$ for $E_7^*$
relative to coset-by-coset decoding. We also discuss further
constant-factor improvements available from recent refinements of the
$A_n^*$ algorithms, and an open question concerning sort-free
linear-time decoding.
\end{abstract}

\section{Introduction}

\subsection{Motivation}

A central figure of merit of a lattice used for quantization is its
\emph{normalized second moment} $G$, which measures the average squared
quantization error per dimension incurred when a uniformly distributed
source is quantized to the nearest lattice point. In dimensions six and
seven, the best lattice quantizers known are the dual lattices $E_6^*$
and $E_7^*$~\cite[Ch.~2, Table~2.3]{ConwaySloaneBook}: their normalized
second moments, computed exactly by
Worley~\cite{Worley1987,Worley1988}, are the smallest among all known
six- and seven-dimensional lattices. Consequently, whenever an
application calls for vector quantization in these dimensions---examples
include gain/shape and subband quantizers in audio coding, quantization
of short feature or parameter blocks in learned models, and coded
modulation over six- or seven-dimensional constellations---$E_6^*$ and
$E_7^*$ are the natural first candidates.

A theoretically good lattice is useful in practice only if the
\emph{closest-point problem}---given a query $\mathbf{y}$, find
$\mathbf{x}^*=\argmin_{\mathbf{x}\in\Lambda}\|\mathbf{y}-\mathbf{x}\|^2$---can
be solved quickly. This paper presents faster closest-point algorithms
for $E_6^*$ and $E_7^*$; the same technique also covers $E_6$ and
$E_7$.

\subsection{Prior work}

Conway and Sloane~\cite{ConwaySloane1982} gave the classical fast
closest-point algorithms for the root lattices $\Z^n$, $A_n$, $D_n$,
$E_7$, $E_8$ and their duals. Their central device is \emph{coset
decoding}: if a lattice of interest can be written as a finite union
\[
    \Lambda'=\bigcup_{i=0}^{d-1}(\mathbf{r}_i+\Lambda)
\]
of shifted copies (cosets) of a lattice $\Lambda$ for which a
closest-point algorithm is available, then $\Lambda'$ is decoded by
quantizing $\mathbf{y}-\mathbf{r}_i$ to $\Lambda$ for each $i$, shifting
each result back by $\mathbf{r}_i$, and keeping the closest of the $d$
candidates. For example, $E_7^*$ is the union of $d=4$ cosets of $A_7$
and is decoded by four $A_7$ quantizations plus four distance
computations. Soft-decision variants of these decoders appear
in~\cite{ConwaySloane1986}. Notably, \cite{ConwaySloane1982} does not
treat $E_6$ and $E_6^*$.

That gap was closed by Takizawa, Yagi, and Kawabata
(TYK)~\cite{TYK2010}, who also generalized the optimality proofs of the
Conway--Sloane quantizers to $\ell_p$ norms. TYK showed that, in the
standard coordinates $E_6\subset\R^8$, the lattice splits across two
orthogonal coordinate blocks into a two-part union of direct products of
copies of $A_1$ and $A_5$ and their translates, and that $E_6^*$ is a
union of three cosets of $E_6$. The resulting decoder performs
$3\times2=6$ product-coset decodings and keeps the best.

In parallel, the closest-point problem for the lattices $A_n^*$ has seen
significant algorithmic progress. Conway and Sloane's original method
requires $O(n^2\log n)$ time, improved to $O(n^2)$
in~\cite{ConwaySloane1986}. Clarkson~\cite{Clarkson1999} reduced this to
$O(n\log n)$. McKilliam, Clarkson, and Quinn (MCQ)~\cite{MCQ} obtained a
particularly simple $O(n\log n)$ method, built around a one-dimensional
sweep over a chain of $n+1$ candidate points; we review it in
Section~\ref{sec:mcq}, as it is the foundation of this paper. McKilliam,
Clarkson, Smith, and Quinn (MCSQ)~\cite{MCSQ} subsequently replaced the
sort inside MCQ by bucketing, achieving $O(n)$ time, and a recent
refinement by the author~\cite{fanstarPaper,fanstar} reduces the
constant factors further; we return to these in
Section~\ref{sec:improvements}.

The contribution of this paper is to connect the two lines of work. We
show that the candidate chain of MCQ computes, in a single pass, the
closest points of \emph{all} cosets of $A_n$ inside $A_n^*$---precisely
the quantities that coset decoders for $E_6^*$ and $E_7^*$ compute one
at a time. Combined with a glue-vector reformulation of the TYK
decompositions, this collapses the four coset decodings of $E_7^*$, and
the six of $E_6^*$, into one sorted sweep per coordinate block.

\section{Preliminaries}
\label{sec:prelim}

\subsection{The lattices $A_n$, $A_n^*$, and glue vectors}
\label{sec:glue}

The root lattice $A_n$ is most conveniently described in $N=n+1$
ambient coordinates:
\[
    A_n=\{\mathbf{k}\in\Z^{N}:\one^T\mathbf{k}=0\},
\]
the integer points of the hyperplane
$H=\{\mathbf{x}\in\R^N:\one^T\mathbf{x}=0\}$, where $\one$ is the
all-ones vector. Its dual lattice $A_n^*$ (in the standard
normalization) is a denser lattice in the same hyperplane. It decomposes
as a disjoint union of $N$ shifted copies of $A_n$,
\begin{equation}
    A_n^*=\bigcup_{c=0}^{n}\bigl([c]+A_n\bigr),
    \qquad
    [c]=\Bigl(
        \bigl(\tfrac{c}{N}\bigr)^{N-c},\,
        \bigl(-\tfrac{N-c}{N}\bigr)^{c}
    \Bigr),
    \label{eq:glue}
\end{equation}
where $\alpha^i$ denotes the number $\alpha$ repeated $i$ times. The
shift vectors $[c]$ are called \emph{glue vectors}, and $c$ is the
\emph{glue index} or class of the coset~\cite[Ch.~4]{ConwaySloaneBook}.
Two properties make glue indices convenient. First, they add modulo $N$:
$[c]+[c']$ is congruent to $[c+c'\bmod N]$ modulo $A_n$. Second, class
membership is trivial to read off: a point of $A_n^*$ lies in class $c$
exactly when every one of its coordinates is congruent to $c/N$ modulo
$1$. For instance, the point
$(\tfrac13,\tfrac13,-\tfrac23)\in A_2^*$ has all coordinates congruent
to $1/3\pmod 1$ and lies in class $c=1$ of $A_2^*$ ($N=3$).

The smallest case, $A_1^*$, will appear as a building block below: it
lives on the line $x_1+x_2=0$ in $\R^2$, consists of the points
$(u,-u)$ with $u\in\tfrac12\Z$, and has two glue classes---class $0$
(integer $u$) and class $1$ (half-integer $u$), with glue vector
$[1]=(-\tfrac12,\tfrac12)$.

\subsection{The MCQ algorithm for $A_n^*$}
\label{sec:mcq}

We now review the closest-point algorithm of McKilliam, Clarkson, and
Quinn~\cite{MCQ} in a self-contained form, since our main results are
built directly on its primitives.

Every point of $A_n^*$ can be written as the projection of an integer
vector onto the hyperplane $H$: with $Q=I-\one\one^T/N$ denoting the
orthogonal projector onto $H$,
\[
    A_n^*=Q\,\Z^N
    =\Bigl\{\mathbf{k}-\frac{\one^T\mathbf{k}}{N}\one
       \;:\;\mathbf{k}\in\Z^N\Bigr\}.
\]
This representation is not unique---adding $\one$ to $\mathbf{k}$ leaves
$Q\mathbf{k}$ unchanged---but it is algorithmically convenient: to find
the closest point of $A_n^*$ to a query $\mathbf{y}\in H$, it suffices
to find a suitable integer representative $\mathbf{k}$.

\paragraph{Candidates by shifted rounding.}
Let $\round(\cdot)$ denote coordinatewise rounding to the nearest
integer. A natural family of candidate representatives is obtained by
rounding shifted copies of the query:
\[
    f(\lambda)=\round(\mathbf{y}+\lambda\one),
    \qquad \lambda\in\R .
\]
Since $f(\lambda+1)=f(\lambda)+\one$ and $Q\one=\mathbf{0}$, the
projected point $Qf(\lambda)$ is periodic in $\lambda$ with period $1$,
so only $\lambda\in[0,1)$ matters. McKilliam et al.\ prove that a
closest point of $A_n^*$ to $\mathbf{y}$ is always of the form
$Qf(\lambda)$ for some $\lambda$~\cite{MCQ}; intuitively, sliding
$\lambda$ trades off how many coordinates round up, which is exactly
the freedom that distinguishes the points of $A_n^*$ from plain
rounding.

\paragraph{Residuals, thresholds, and the candidate chain.}
Let $\mathbf{k}_0=\round(\mathbf{y})$ and let
\[
    z_j=y_j-k_{0,j}\in[-\tfrac12,\tfrac12)
\]
be the centered rounding residuals. As $\lambda$ grows from $0$ to $1$,
coordinate $j$ of $f(\lambda)$ increases by one at the moment $\lambda$
crosses the \emph{threshold}
\[
    \tau_j=\tfrac12-z_j\in(0,1] .
\]
Coordinates with large residuals (values that were ``almost rounded
up'') flip first. Sorting the thresholds in increasing order---
equivalently, sorting the residuals in decreasing order---and flipping
one coordinate at a time produces the \emph{candidate chain}
\[
    \mathbf{k}_0,\ \mathbf{k}_1,\ \dots,\ \mathbf{k}_n,
\]
where $\mathbf{k}_i$ equals $\mathbf{k}_{i-1}$ with one more coordinate
increased by one, and $\one^T\mathbf{k}_i=s+i$ with
$s=\one^T\mathbf{k}_0$. One of the projected points $Q\mathbf{k}_i$ is a
closest point of $A_n^*$ to $\mathbf{y}$.

\paragraph{Evaluating the candidates in constant time each.}
For $\mathbf{z}_i=\mathbf{y}-\mathbf{k}_i$, a direct computation using
$\mathbf{y}\in H$ gives the squared distance to the projected candidate:
\begin{equation}
    d_i=\|\mathbf{y}-Q\mathbf{k}_i\|^2
       =\beta_i-\frac{\alpha_i^2}{N},
    \qquad
    \alpha_i=\one^T\mathbf{z}_i,\quad
    \beta_i=\mathbf{z}_i^T\mathbf{z}_i .
    \label{eq:di}
\end{equation}
When candidate $i$ is formed from candidate $i-1$ by incrementing
coordinate $(i)$, the residual of that coordinate drops by one and all
others are untouched, so
\begin{equation}
    \alpha_i=\alpha_{i-1}-1,
    \qquad
    \beta_i=\beta_{i-1}-2z_{(i)}+1 .
    \label{eq:rec}
\end{equation}
Thus, after one $O(n\log n)$ sort of the thresholds, the entire chain
can be swept in $O(n)$ time, evaluating \eqref{eq:di} via
\eqref{eq:rec} in $O(1)$ per candidate and keeping the minimum. This is
the MCQ algorithm.

\section{Main results}
\label{sec:main}

Our results combine two ingredients: a reformulation of the TYK
decompositions of $E_6$, $E_6^*$, $E_7$, $E_7^*$ in the glue language of
Section~\ref{sec:glue}, and an observation about the MCQ chain that
appears not to have been exploited before.

\subsection{The TYK decompositions in glue form}
\label{sec:glueform}

TYK~\cite{TYK2010} work in the standard embedding
\[
    E_6=\{\mathbf{x}\in E_8: x_1+\dots+x_8=x_1+x_8=0\}\subset\R^8,
\]
which splits the coordinates into two orthogonal blocks: the pair
$(x_1,x_8)$ and the six-tuple $(x_2,\dots,x_7)$. Writing $\tilde A_1$
and $\tilde A_5$ for the copies of $A_1$ and $A_5$ supported on these
blocks, TYK's Lemma~1 states
\begin{equation}
    E_6
    =
    \bigl(\tilde A_1\otimes\tilde A_5\bigr)
    \;\cup\;
    \bigl((\tilde{\mathbf{d}}^{(1)}+\tilde A_1)\otimes
          (\tilde{\mathbf{d}}^{(2)}+\tilde A_5)\bigr),
    \label{eq:tyk6}
\end{equation}
with $\tilde{\mathbf{d}}^{(1)}=(-\tfrac12,\tfrac12)$ and
$\tilde{\mathbf{d}}^{(2)}=\bigl((-\tfrac12)^3,(\tfrac12)^3\bigr)$, and
$E_6^*=\bigcup_{i=0}^{2}(\mathbf{r}_i+E_6)$ with
$\mathbf{r}_1=\bigl(0,(-\tfrac23)^2,(\tfrac13)^4,0\bigr)$ and
$\mathbf{r}_2=-\mathbf{r}_1$. Similarly, in the sum-zero hyperplane of
$\R^8$,
\begin{equation}
    E_7=A_7\cup
    \bigl(\bigl((-\tfrac12)^4,(\tfrac12)^4\bigr)+A_7\bigr),
    \qquad
    E_7^*=\bigcup_{i=0}^{3}(\mathbf{r}'_i+A_7),
    \label{eq:tyk7}
\end{equation}
with $\mathbf{r}'_1=\bigl((-\tfrac34)^2,(\tfrac14)^6\bigr)$,
$\mathbf{r}'_2=\bigl((-\tfrac12)^4,(\tfrac12)^4\bigr)$,
$\mathbf{r}'_3=\bigl((-\tfrac14)^6,(\tfrac34)^2\bigr)$.

Comparing these shift vectors with \eqref{eq:glue} identifies every one
of them, up to coordinate permutation and lattice translation, as a glue
vector: $\tilde{\mathbf{d}}^{(1)}$ is the glue vector $[1]$ of $A_1$;
$\tilde{\mathbf{d}}^{(2)}$ is $[3]$ of $A_5$; the $A_5$-block parts of
$\mathbf{r}_1,\mathbf{r}_2$ are $[2]$ and $[4]$ of $A_5$; and
$\mathbf{r}'_1,\mathbf{r}'_2,\mathbf{r}'_3$ are $[2],[4],[6]$ of $A_7$.
Using additivity of glue indices (for example, the second part of
\eqref{eq:tyk6} shifted by $\mathbf{r}_1$ carries $A_5$-glue
$[3]+[2]=[5]$), the decompositions
\eqref{eq:tyk6}--\eqref{eq:tyk7} take the following form.

\begin{lemma}[glue form of the TYK decompositions]
\label{lem:glue}
With glue indices $j\in\{0,1\}$ for $A_1^*$ and $c$ for $A_5^*$ or
$A_7^*$,
\begin{align}
    E_7^* &= \bigcup_{c\in\{0,2,4,6\}} \bigl([c]+A_7\bigr),
    \qquad\quad
    E_7 = \bigcup_{c\in\{0,4\}} \bigl([c]+A_7\bigr),
    \label{eq:e7glue}\\[2pt]
    E_6^* &= \bigl\{(\mathbf{a},\mathbf{b})\in A_1^*\oplus A_5^*
              \;:\; j(\mathbf{a})\equiv c(\mathbf{b}) \pmod 2\bigr\},
    \label{eq:e6glue}\\[2pt]
    E_6 &= \bigl\{(\mathbf{a},\mathbf{b})\in A_1^*\oplus A_5^*
              \;:\; (j(\mathbf{a}),c(\mathbf{b}))\in\{(0,0),(1,3)\}\bigr\}.
    \notag
\end{align}
\end{lemma}

In words: $E_7^*$ is the union of the \emph{even} glue classes of
$A_7^*$, and $E_6^*$ is a \emph{parity-matched} sublattice of the direct
sum $A_1^*\oplus A_5^*$---a ``checkerboard'' coupling of the two blocks,
structurally analogous to the familiar description
$E_8=D_8\cup(\tfrac12^8+D_8)$. Note that all six glue classes of $A_5$
appear in \eqref{eq:e6glue}: the union of the blockwise search spaces is
the \emph{entire} dual $A_1^*\oplus A_5^*$, restricted only by the
parity coupling.

\subsection{One sweep decodes all cosets}
\label{sec:sweep}

Coset decoding treats each class in
\eqref{eq:e7glue}--\eqref{eq:e6glue} as a separate quantization problem.
The key observation of this paper is that the MCQ candidate chain
already contains the answers to all of them.

\begin{proposition}[per-coset optimality of the MCQ chain]
\label{prop:coset}
Let $\mathbf{y}\in H$ be a query for $A_n^*$ and let
$\mathbf{k}_0,\dots,\mathbf{k}_n$ be the candidate chain of
Section~\ref{sec:mcq}, built from exactly sorted thresholds. Then the
chain visits each glue class of $A_n$ inside $A_n^*$ exactly once, and
for every $i$ the point $Q\mathbf{k}_i$ is a closest point to
$\mathbf{y}$ within its class, namely class $(s+i)\bmod N$.
\end{proposition}

\begin{proof}
Since $\one^T\mathbf{k}_i=s+i$ for $i=0,\dots,n$, the coordinate sums
exhaust all residues modulo $N$; and $Q\mathbf{k}$, $Q\mathbf{k}'$ lie
in the same glue class if and only if
$\one^T\mathbf{k}\equiv\one^T\mathbf{k}'\pmod N$. This gives the first
claim. For the second, fix a class and a representative sum $m=s+i$.
Every point of the class has a representative with
$\one^T\mathbf{k}=m$ exactly, because adding $\one$ to $\mathbf{k}$
leaves $Q\mathbf{k}$ unchanged while shifting the sum by $N$. On the set
$\{\mathbf{k}:\one^T\mathbf{k}=m\}$, the objective \eqref{eq:di} has a
constant second term, so minimizing it is the same as minimizing
$\|\mathbf{y}-\mathbf{k}\|^2$ subject to the fixed coordinate sum. The
solution of this classical constrained rounding problem is to round
$\mathbf{y}$ and then increment the $i$ coordinates with the largest
residuals (see, e.g., the $A_n$ quantizer of~\cite{ConwaySloane1982} and
its optimality proof in~\cite{TYK2010})---which is exactly
$\mathbf{k}_i$, since the chain increments coordinates in decreasing
order of residual.
\end{proof}

\begin{corollary}
\label{cor:allcosets}
One sweep along the sorted MCQ chain returns, simultaneously and at no
extra asymptotic cost, a closest point of every glue coset of $A_n$ to
$\mathbf{y}$, together with its squared distance $d_i$.
\end{corollary}

Corollary~\ref{cor:allcosets} converts Lemma~\ref{lem:glue} directly
into single-sweep decoders. For $E_7^*$, instead of taking the global
minimum of $d_i$ along the chain, the sweep restricts the minimization
to candidates whose class $(s+i)\bmod 8$ is even---i.e., to indices $i$
with $i\equiv s\pmod 2$. For $E_6^*$, one sweep per block yields
per-class minima: $D_a[j]$ for the two classes of the (trivial) $A_1$
block, and $D_b[c]$, $c=0,\dots,5$, for the $A_5$ block. Because the two
blocks occupy disjoint coordinates, squared distances add, and by
\eqref{eq:e6glue} the coupling reduces to the $O(1)$ comparison
\begin{equation}
    \min\Bigl(
        D_a[0]+\min_{c\,\mathrm{even}}D_b[c],\;\;
        D_a[1]+\min_{c\,\mathrm{odd}}D_b[c]
    \Bigr).
    \label{eq:combine}
\end{equation}
Since the MCQ objective $d_i$ in \eqref{eq:di} is the true squared
distance, the quantities being added in \eqref{eq:combine} are directly
comparable; no renormalization between blocks is needed.

\subsection{The algorithms}
\label{sec:algs}

Algorithms~\ref{alg:e7} and~\ref{alg:e6} state the resulting procedures.
Both are direct instantiations of the MCQ sweep with class-restricted
minimization; the reader familiar with~\cite{MCQ} will recognize every
step. For the fixed block sizes arising here ($N=8$ and $N=6$), the
sorting step is a constant-size sorting network ($19$ and $12$
compare-exchange operations, respectively), so the asymptotic
distinction between $O(n\log n)$ and $O(n)$ is immaterial; we return to
it in Section~\ref{sec:improvements}.

\begin{algorithm}[t]
\caption{Closest point in $E_7^*$. Input
$\mathbf{y}\in\R^8$ with $\one^T\mathbf{y}=0$; $N=8$.}
\label{alg:e7}
\begin{algorithmic}[1]
\State $\mathbf{k}\gets\round(\mathbf{y})$;\quad
       $\mathbf{z}\gets\mathbf{y}-\mathbf{k}$;\quad
       $s\gets\one^T\mathbf{k}$
\State $\alpha\gets\one^T\mathbf{z}$;\quad
       $\beta\gets\mathbf{z}^T\mathbf{z}$
\State sort indices so that $z_{(1)}\ge z_{(2)}\ge\dots\ge z_{(N)}$
       \Comment{equivalently, thresholds $\tau_{(i)}$ ascending}
\State $D\gets+\infty$;\quad $m\gets 0$
\If{$s\equiv 0\pmod 2$}
    \State $D\gets\beta-\alpha^2/N$
    \Comment{candidate $0$ lies in class $s\bmod N$}
\EndIf
\For{$i\gets 1$ \textbf{to} $N-1$}
    \State $\alpha\gets\alpha-1$;\quad
           $\beta\gets\beta-2z_{(i)}+1$
    \If{$s+i\equiv 0\pmod 2$}
        \Comment{even glue classes only, cf.\ \eqref{eq:e7glue}}
        \State $h\gets\beta-\alpha^2/N$
        \If{$h<D$}  $D\gets h$;\quad $m\gets i$ \EndIf
    \EndIf
\EndFor
\State $k_{(j)}\gets k_{(j)}+1$ for $j=1,\dots,m$
\State $\mu\gets(\one^T\mathbf{k})/N$;\quad
       \Return $\mathbf{k}-\mu\one$
\end{algorithmic}
\end{algorithm}

\begin{algorithm}[t]
\caption{Closest point in $E_6^*$. Input
$\mathbf{y}\in\R^8$ with $y_1+y_8=0$ and $y_2+\dots+y_7=0$.}
\label{alg:e6}
\begin{algorithmic}[1]
\State \emph{($A_1$ block, coordinates $(y_1,y_8)=(t,-t)$; closed form)}
\State $u_0\gets\round(t)$;\quad $D_a[0]\gets 2(t-u_0)^2$
\State $u_1\gets\round(t-\tfrac12)+\tfrac12$;\quad
       $D_a[1]\gets 2(t-u_1)^2$
\State \emph{($A_5$ block, coordinates
       $\mathbf{y}^{(b)}=(y_2,\dots,y_7)$; $N=6$)}
\State run the sweep of Algorithm~\ref{alg:e7} on $\mathbf{y}^{(b)}$,
       but track \emph{two} running minima: $D_b[\mathrm{ev}]$ over
       candidates with $(s+i)$ even and $D_b[\mathrm{od}]$ over
       candidates with $(s+i)$ odd, with winning indices
       $m_{\mathrm{ev}}$, $m_{\mathrm{od}}$
\State \emph{(combine, cf.\ \eqref{eq:e6glue}--\eqref{eq:combine})}
\If{$D_a[0]+D_b[\mathrm{ev}] \le D_a[1]+D_b[\mathrm{od}]$}
    \State $j^*\gets 0$;\quad $m^*\gets m_{\mathrm{ev}}$
\Else
    \State $j^*\gets 1$;\quad $m^*\gets m_{\mathrm{od}}$
\EndIf
\State reconstruct the $A_5$-block point from $m^*$ as in
       Algorithm~\ref{alg:e7}; set the $A_1$-block point to
       $(u_{j^*},-u_{j^*})$
\State \Return the concatenation of the two block points
\end{algorithmic}
\end{algorithm}

Decoding $E_7$ is identical to Algorithm~\ref{alg:e7} with the
restriction $(s+i)\bmod 8\in\{0,4\}$, and decoding $E_6$ is identical to
Algorithm~\ref{alg:e6} with classes $c=0$ and $c=3$ tracked in place of
the even/odd minima. Decoding $A_7^*$ or $A_5^*$ itself corresponds to
no restriction at all---the three problems differ only in which chain
candidates participate in the minimization.

In all cases, the work amounts to a single MCQ-style quantization of
each coordinate block---one block of size $8$ for $E_7^*$, and blocks of
sizes $6$ and $2$ for $E_6^*$---plus a constant amount of bookkeeping.
No shifted copies of the input are formed, and no explicit
$\ell_2$ distances between $8$-dimensional vectors are ever computed.

\section{Complexity comparison}
\label{sec:complexity}

Coset decoding of $E_7^*$ in the style of~\cite{ConwaySloane1982}
performs four $A_7$ quantizations of shifted inputs and four explicit
distance evaluations; the TYK decoder for $E_6^*$~\cite{TYK2010}
performs six product-coset decodings (each one $A_5$-coset and one
$A_1$-coset quantization) and six distance evaluations. The proposed
decoders perform one sweep per block. Table~\ref{tab:counts} lists rough
operation counts; exact constants depend on implementation choices (in
particular, each coset decoding also requires forming the shifted input
$\mathbf{y}-\mathbf{r}_i$ and shifting the output back, counted under
``shift add/subtract passes'').

\begin{table}[t]
\centering
\caption{Rough operation counts per query for closest-point computation
in $E_6^*$ and $E_7^*$. ``Sorts'' counts residual orderings as
(number)$\times$(size); sizes $6$ and $8$ admit fixed sorting networks
of $12$ and $19$ compare-exchange operations, respectively. Distance
evaluations are explicit $8$-dimensional squared-distance computations
($\approx 24$ flops each). Chain objective evaluations are the $O(1)$
steps \eqref{eq:rec}--\eqref{eq:di}.}
\label{tab:counts}
\small
\begin{tabular}{l cc cc}
\toprule
 & \multicolumn{2}{c}{$E_6^*$} & \multicolumn{2}{c}{$E_7^*$} \\
\cmidrule(lr){2-3}\cmidrule(lr){4-5}
Operation & TYK~\cite{TYK2010} & Proposed & Coset~\cite{ConwaySloane1982} & Proposed \\
\midrule
Coordinate roundings          & $48$ & $8$ & $32$ & $8$  \\
Sorts of residuals            & $6\times 6$ & $1\times 6$ & $4\times 8$ & $1\times 8$ \\
Chain objective evaluations   & --   & $6$ & --   & $8$  \\
Explicit distance evaluations & $6$  & $0$ & $4$  & $0$  \\
Shift add/subtract passes     & $12$ & $0$ & $8$  & $0$  \\
Output/reconstruction passes  & $1$  & $1$ & $1$  & $1$  \\
\bottomrule
\end{tabular}
\end{table}

Summing the dominant terms, the proposed method reduces the work by
roughly a factor of $4$--$6$ for $E_6^*$ and $3$--$4$ for $E_7^*$. Since
the block sweeps are essentially $A_5^*$ and $A_7^*$ quantizations, the
resulting $E_6^*$ and $E_7^*$ decoders are expected to be comparable in
speed to the standard $E_8$ decoder, removing the main computational
objection to using the best known six- and seven-dimensional quantizers
in practice.

\section{Additional improvements and future work}
\label{sec:improvements}

The algorithms above are built on MCQ because its exactly sorted chain
is what makes Proposition~\ref{prop:coset} unconditional. Two further
lines of improvement are available, one immediate and one requiring new
theory.

\paragraph{Constant-factor refinements.}
For $A_n^*$, MCSQ~\cite{MCSQ} and the author's
refinement~\cite{fanstarPaper} improve on MCQ in ways that are largely
orthogonal to the results of this paper and transfer to the block
sweeps. The refinements of~\cite{fanstarPaper} include: dropping the
constant $\beta_0=\mathbf{z}^T\mathbf{z}$ from the objective (in the
$E_6^*$ combination \eqref{eq:combine} the per-block constants appear in
both totals and cancel, so reduced objectives remain valid); scaling the
reduced objective so that its quadratic term, the surplus
$\alpha_0=-s$, and the flip counts are exact integers (for the
two-block case the objectives must be rescaled to a common factor,
$\mathrm{lcm}(2,6)=6$); fusing the rounding and threshold-computation
passes so that residuals are consumed as they are produced; and
reconstructing the output branchlessly. When the input coordinates are
rationals with a common denominator---histograms, empirical types---the
entire computation becomes exact integer arithmetic. Since the block
sizes here are the fixed constants $6$ and $8$, fully unrolled
implementations with hard-coded sorting networks, in the style of the
small-$n$ routines of~\cite{fanstar}, are the natural realization.

\paragraph{Sort-free decoding: an open question.}
MCSQ~\cite{MCSQ} eliminates the sort inside MCQ by placing the
thresholds into $N$ buckets of width $1/N$ and sweeping the buckets in
order, achieving $O(n)$ total time. Its correctness rests on an interval
lemma: the optimal candidate $Qf(\lambda)$ of $A_n^*$ is constant on a
$\lambda$-interval of length at least $1/N$, so the coarse bucket
ordering cannot miss it. That guarantee, however, concerns only the
\emph{global} optimum. The per-coset optimality of
Proposition~\ref{prop:coset} genuinely requires exact ordering: for a
fixed coset, the optimal candidate occupies a $\lambda$-interval equal
to a gap between consecutive sorted thresholds, which can be arbitrarily
short. For the fixed small blocks of $E_6^*$ and $E_7^*$ this is
irrelevant, but it leaves an interesting theoretical question. The
constrained search spaces arising here are themselves lattices---the
even-glue union in \eqref{eq:e7glue} is an index-$2$ sublattice of
$A_7^*$, and \eqref{eq:e6glue} is an index-$2$ sublattice of
$A_1^*\oplus A_5^*$---so one may ask whether an analogous interval lemma
holds for glue-constrained decoding (plausibly with the bound $1/N$
replaced by a constant reflecting the sparser sublattice). A positive
answer would yield sort-free linear-time decoders for the whole family
of glue constructions over $A_n$ blocks in growing dimension. More
broadly, the template of Section~\ref{sec:sweep}---sweep once per
block, harvest per-class minima, then solve a small matching problem
over the glue group---applies verbatim to other lattices assembled by
gluing $A_n$ blocks, and cataloguing the lattices reachable this way is
a natural next step. An experimental evaluation within the open-source
project \texttt{fanstar}~\cite{fanstar} is also planned.

\section{Conclusions}

We presented faster closest-point algorithms for $E_6^*$ and $E_7^*$,
the best lattice quantizers known in dimensions six and
seven~\cite{ConwaySloaneBook,Worley1987,Worley1988}. The construction
rests on two observations. First, the TYK decompositions of $E_6$,
$E_6^*$, $E_7$, and $E_7^*$, restated in glue-vector language, exhibit
these lattices as glue-constrained subsets of $A_1^*\oplus A_5^*$ and
$A_7^*$: $E_7^*$ is the union of the even glue classes of $A_7^*$, and
$E_6^*$ is a parity-matched sublattice of $A_1^*\oplus A_5^*$. Second,
the candidate chain of the MCQ algorithm visits every glue class of
$A_n$ exactly once and is optimal within each class, so a single sorted
sweep yields all per-coset minima simultaneously. Together, these
replace the four to six separate coset decodings and explicit distance
computations of prior methods by one sweep per coordinate block,
reducing operation counts by roughly $3$--$6\times$.


\begin{thebibliography}{9}

\bibitem{ConwaySloane1982}
J.~H. Conway and N.~J.~A. Sloane,
``Fast quantizing and decoding algorithms for lattice quantizers and
codes,''
\emph{IEEE Trans. Inf. Theory}, vol.~28, no.~2, pp.~227--232, Mar.~1982.

\bibitem{ConwaySloane1986}
J.~H. Conway and N.~J.~A. Sloane,
``Soft decoding techniques for codes and lattices, including the Golay
code and the Leech lattice,''
\emph{IEEE Trans. Inf. Theory}, vol.~32, no.~1, pp.~41--50, Jan.~1986.

\bibitem{ConwaySloaneBook}
J.~H. Conway and N.~J.~A. Sloane,
\emph{Sphere Packings, Lattices and Groups}, 3rd ed.
New York: Springer-Verlag, 1998.

\bibitem{Worley1987}
R.~T. Worley,
``The Voronoi region of $E_6^*$,''
\emph{J. Austral. Math. Soc., Ser.~A}, vol.~43, no.~2, pp.~268--278,
1987.

\bibitem{Worley1988}
R.~T. Worley,
``The Voronoi region of $E_7^*$,''
\emph{SIAM J. Discrete Math.}, vol.~1, no.~1, pp.~134--141, 1988.

\bibitem{Clarkson1999}
I.~V.~L. Clarkson,
``An algorithm to compute a nearest point in the lattice $A_n^{*}$,''
in \emph{Applied Algebra, Algebraic Algorithms and Error-Correcting
Codes} (Lecture Notes in Computer Science), vol.~1719.
Springer, 1999, pp.~104--120.

\bibitem{MCQ}
R.~G. McKilliam, I.~V.~L. Clarkson, and B.~G. Quinn,
``An algorithm to compute the nearest point in the lattice $A_n^{*}$,''
\emph{IEEE Trans. Inf. Theory}, vol.~54, no.~9, pp.~4378--4381,
Sep.~2008.

\bibitem{MCSQ}
R.~G. McKilliam, I.~V.~L. Clarkson, W.~D. Smith, and B.~G. Quinn,
``A linear-time nearest point algorithm for the lattice $A_n^{*}$,''
in \emph{Proc. Int. Symp. Information Theory and its Applications
(ISITA'08)}, Auckland, New Zealand, Dec.~2008, pp.~1--5.

\bibitem{TYK2010}
K.~Takizawa, H.~Yagi, and T.~Kawabata,
``Closest point algorithms with $\ell_p$ norm for root lattices,''
in \emph{Proc. IEEE Int. Symp. Information Theory (ISIT'10)},
Austin, TX, June 2010, pp.~1042--1046.

\bibitem{fanstarPaper}
Y.~Reznik,
``Faster closest-point algorithms for the $A_n^{*}$ lattices,''
preprint, 2026.

\bibitem{fanstar}
Y.~Reznik,
\emph{fanstar: Faster closest-point algorithms for the $A_n^{*}$ lattices},
open-source software.
\url{https://github.com/yuriy-a-reznik/fanstar}

\end{thebibliography}
\end{document}